\documentclass[12pt]{article}
\usepackage{bm}
\usepackage{amsthm}
\usepackage{amsfonts}
\usepackage{amssymb}
\usepackage{amsmath}
\usepackage{multirow}

\makeatletter
\theoremstyle{plain}
\newtheorem{thm}{\protect\theoremname}
  \theoremstyle{plain}
  \newtheorem{cor}[thm]{\protect\corollaryname}
  \theoremstyle{plain}
  \newtheorem{prop}[thm]{\protect\propositionname}
  \theoremstyle{definition}
  \newtheorem{defn}[thm]{\protect\definitionname}
  \theoremstyle{plain}
  \newtheorem{lem}[thm]{\protect\lemmaname}

\makeatletter
\DeclareMathOperator{\dom}{dom}

\DeclareMathOperator{\si}{subIso}

\DeclareMathOperator{\size}{size}
\DeclareMathOperator{\ar}{ar}
\newcommand{\mhyphen}{\mbox{\it -}}
\newcommand{\tup}[1]{\langle #1 \rangle}

\makeatother

\usepackage[english]{babel}
  \providecommand{\corollaryname}{Corollary}
  \providecommand{\definitionname}{Definition}
  \providecommand{\lemmaname}{Lemma}
  \providecommand{\propositionname}{Proposition}
\providecommand{\theoremname}{Theorem}

\title{The Complexity of Definability\\ by Open First-Order Formulas}

\date{}

\author{
\begin{tabular}{cc}
Carlos Areces &
Miguel Campercholi \\
\texttt{\small carlos.areces@gmail.com} & 
\texttt{\small mcampercholi@gmail.com}\\[1em]
Daniel Penazzi & 
Pablo Ventura \\
\texttt{\small danielpenazzi@gmail.com}  &
\texttt{\small pablogventura@gmail.com} \\[1em]
\multicolumn{2}{c}{Universidad Nacional de C\'ordoba and CONICET}\\
\multicolumn{2}{c}{C\'ordoba, Argentina}
\end{tabular}
}

\begin{document}

\maketitle

\begin{abstract}
  In this article we formally define and investigate the computational
  complexity of the Definability Problem for open first-order formulas
  (i.e., quantifier free first-order formulas) with equality. Given a
  logic $\bm{\mathcal{L}}$, the $\bm{\mathcal{L}}$-Definability
  Problem for finite structures takes as input a finite structure
  $\bm{A}$ and a target relation $T$ over the domain of $\bm{A}$, and
  determines whether there is a formula of $\bm{\mathcal{L}}$ whose
  interpretation in $\bm{A}$ coincides with $T$. We show that the
  complexity of this problem for open first-order formulas (open
  definability, for short) is coNP-complete. We also investigate the
  parametric complexity of the problem, and prove that if the size and
  the arity of the target relation $T$ are taken as parameters then
  open definability is $\mathrm{coW}[1]$-complete for every vocabulary
  $\tau$ with at least one, at least binary, relation.
\end{abstract}

\section{Introduction}

Arguably, any attempt to provide a logic $\bm{\mathcal{L}}$ with a
formal semantics starts with the definition of a function that, given
a suitable structure $\bm{A}$ for $\bm{\mathcal{L}}$ and a formula
$\varphi$ in $\bm{\mathcal{L}}$, returns the \emph{extension} of
$\varphi$ in $\bm{A}$.  Usually, this extension is a set of tuples
built from elements in $\bm{A}$. These extensions, also called
\emph{definable sets,} are the elements that will be referred by the
formulas of $\bm{\mathcal{L}}$ in a given structure, and in that
sense, define the expressivity of $\bm{\mathcal{L}}$. The definable
sets of $\bm{A}$ are the only objects that $\bm{\mathcal{L}}$ can
\emph{see}. For that reason, definable sets are one of the central
objects studied by Model Theory. It is usually an interesting question
to investigate, given a logic $\bm{\mathcal{L}}$, which are the
definable sets of $\bm{\mathcal{L}}$ over a given structure $\bm{A}$,
or, more concretely, whether a particular set of tuples is a definable
set of $\bm{\mathcal{L}}$ over $\bm{A}$. This is what we call the
\emph{Definability Problem for $\bm{\mathcal{L}}$ over $\bm{A}$}.

In this article we investigate the computational complexity of the
definability problem for open first-order formulas --i.e.,
quantifier free first-order formulas-- with equality over a
relational vocabulary (open-definability, for short).

One of the main goals of Computational Logic is to understand the
computational complexity of different problems for different logics.
Classically, one of the most investigated inference problems is
\emph{Satisfiability} (SAT, for short): given a formula $\varphi$ from
a given logic $\bm{\mathcal{L}}$ decide whether there exists a structure
that makes $\varphi$ true. In recent years, and motivated by concrete
applications, other reasoning problems have sparkled interest. A well
known example is the \emph{Model Checking Problem} (MC, for short)
used in software verification to check that a given property $P$
(expressed as a formula in the verification language) holds in a given
formal representation $S$ of the system (see,
e.g.,~\cite{clar:mode99,bera:syst10}). From a more general
perspective, MC can be defined as follows: given a structure $\bm{A}$,
and a formula $\varphi$ decide which is the extension
$T$ of $\varphi$ in $\bm{A}$. From that perspective,
the definability problem can be understood as the \emph{inverse}
problem of MC: given a structure $\bm{A}$ and a target set $T$ it asks
whether there is a formula $\varphi$ whose extension
is $T$.  A further example of a reasoning task related to
definability comes from a seemingly unrelated field: computational
linguistics, more specifically, in the subarea of automated language
generation called Generation of Referring Expressions (GRE). The GRE
problem can be intuitively understood as follows: given a context $C$
and an target object $t$ in $C$, generate a grammatically correct
description (in some natural language) that represents $t$,
differentiating it from other objects in $C$, or report failure if
such a description does not exist (see~\cite{krah:comp12} for a survey
on GRE). Most of the work in this area is focused on the content
determination problem (i.e., finding the properties that singles out
the target object) and leaves the actual realization (i.e., expressing
this content as a grammatically correct expression) to standard
techniques. As it is discussed in~\cite{arec:refe08,arec:usin11} the
content realization part of the GRE problem can be understood as the
task that, given a structure $\bm{A}$ that represents the context $C$, and
an object $t$ in the domain of $\bm{A}$ returns a formula $\varphi$ in
a suitable logic $\bm{\mathcal{L}}$ whose extension in $\bm{A}$
coincides with $t$. Of course, this will be possible only if $t$ is
definable for $\bm{\mathcal{L}}$ over $\bm{A}$.

The complexity of the definability problem for a number of logics has
already been investigated. Let \textbf{FO} be first-order logic with
equality in a vocabulary without constant symbols. The computational
complexity of \textbf{FO}-definability was already discussed in
1978~\cite{pare:expr78,banc:comp78}, when a semantic characterization
of the problem, based on automorphisms, placed
\textbf{FO}-definability within coNP. Much more recently,
in~\cite{Aren:exac16}, a polynomial-time algorithm for
\textbf{FO}-definability was given, which uses calls to a
graph-isomorphism subroutine as an oracle. As a consequence,
\textbf{FO}-definability is shown to be inside GI (defined as the set
of all languages that are polynomial-time Turing reducible to the
graph isomorphism problem). The authors also show that the problem is
GI-hard and, hence, GI-complete. Interestingly, Willard showed in
\cite{will:test10}, that the complexity of the definability problem
for the fragment of \textbf{FO} restricted to conjunctive queries (i.e.,
formulas of the form $\exists \bar x \bigwedge_i C_i$, where each
conjunct $C_i$ is atomic) was coNEXPTIME-complete. The complexity upper
bound followed from a semantic characterization of
\textbf{CQ}-definability in terms of polymorphisms given
in~\cite{jeav:howt99}, while the lower bound is proved by an encoding
of a suitable tiling problem. The complexity of definability has been
investigated also for some modal languages: \cite{arec:refe08} shows
that for the basic modal logic $K$, the definability problem is
tractable (i.e., in P); in \cite{arec:usin11} the result is extended
to some fragments of $K$ known as $\mathcal{EL}$ and $\mathcal{EL}^+$.
\cite{figu:size10} discusses the length of the shortest formula
required to define a given target set, proving that for
$\bm{\mathcal{L}} \in \{K, \mathcal{EL}, \mathcal{EL^+}\}$, the lower
bound for the length of a definition is exponential in the size of the
input structure. More precisely, it is shown that there are structures
$G_1$, $G_2$, \ldots such that for every $i$, the size of $G_i$ is
linear in $i$ but the size of the shortest definition for some element
in $G_i$ is bounded from below by a function which is exponential on
$i$.

The article is structured as follows. After introducing basic
notations and definitions in Section~\ref{sec:basic}, we show that
open-definability is coNP-complete in
Section~\ref{sec:Classical-Complexity}. Section~\ref{sec:param}
discusses the parameterized complexity of the problem. Finally, in
Section~\ref{sec:LargoFormulas} we show that the length of the
shortest open formula that might be required in a definition cannot be
bounded by a polynomial: in some cases, definitions by open formulas
need to be exponentially long.

\section{Preliminaries\label{sec:basic}}

In this section we provide some basic definitions and fix  notation.
We assume basic knowledge of first-order logic. For a detailed account
see, e.g., \cite{ebbi:math}.

We focus on definability by open first-order formulas in a purely
relational first-order vocabulary, i.e., without function or constant
symbols. For a relation symbol $R$ in a vocabulary $\tau$, let
$\ar(R)$ denote the arity of $R$. In what follows, all vocabularies
are assumed to be finite and purely relational. We assume that the
language contains variables from a countable, infinite set
$\textsf{VAR}=\{x_{1},x_{2},\ldots,x_{n},\ldots\}$.  Variables from
$\textsf{VAR}$ are the only \emph{terms} in the language.
\emph{Atomic formulas} are either of the form $v_{i}=v_{j}$ or
$R(\bar{v})$, where $v_{i},v_{j}\in\textsf{VAR}$, $\bar{v}$ is a
sequence of variables in $\textsf{VAR}$ of length $k$ and $R$ is a
relation symbol of arity $k$. \emph{Open formulas} are Boolean
combinations of atomic formulas. We shall often write just formula
instead of open formula. We write $\varphi(v_{1},\ldots,v_{k})$ for an
open formula $\varphi$ whose variables are included in
$\{v_{1},\ldots,v_{k}\}$.

Let $\tau$ be a vocabulary. A \emph{$\tau$-structure} (or \emph{model})
is a pair $\bm{A}=\tup{A,\cdot^{\bm{A}}}$ where $A$ is a non-empty
set (the \emph{domain} or \emph{universe}), and $\cdot^{\bm{A}}$
is an \emph{interpretation function} that assigns to each $k$-ary
relation symbol $R$ in $\tau$ a subset $R^{\bm{A}}$ of $A^{k}$.
If $\bm{A}$ is a structure we write $A$ for its domain and $\cdot^{\bm{A}}$
for its interpretation function. Given a formula $\varphi(v_{1},\ldots,v_{k})$,
and a sequence of elements $\bar{a}=\tup{a_{1},\ldots,a_{k}}\in A^{k}$
we write $\bm{A}\models\varphi[\bar{a}]$ if $\varphi$ is true in
$\bm{A}$ under an assignment that maps $v_{i}$ to $a_{i}$.

We say that a subset $T\subseteq A^{k}$ is \emph{open-definable}
in $\bm{A}$ if there is an open first-order formula $\varphi(x_{1},\dots,x_{k})$
in the vocabulary of $\bm{A}$ such that 
\[
T=\{\bar{a}\in A^{k}:\bm{A}\models\varphi[\bar{a}]\}.
\]
In this article we study the following computational decision problem:

\begin{center}
\begin{tabular}{|lrp{8cm}|} \hline
\multirow{3}{*}{\vspace*{15pt}\framebox{$\mathrm{OpenDef}$}}
& & \\
&\emph{Instance}: & A finite relational structure $\bm{A}$ and a relation $T$ over the domain of $\bm{A}$.\\

& \emph{Question}: & Is $T$ open-definable in $\bm{A}$?\\[1em] \hline
\end{tabular}
\end{center}

Let $f:\dom f\subseteq A\rightarrow A$ be a function. Given
$S\subseteq A^{m}$, we say that $f$ \emph{preserves} $S$ if for all
$\left\langle s_{1},\dots,s_{m}\right\rangle \in S\cap(\dom f)^{m}$ we
have $\left\langle fs_{1},\dots,fs_{m}\right\rangle \in S$. The
function $f$ is a \emph{subisomorphism} (subiso for short) of $\bm{A}$
provided that $f$ is injective, and both $f$ and $f^{-1}$ preserve
$R^{\bm{A}}$ for each $R\in\tau$. (Note that a subiso of $\bm{A}$ is
exactly an isomorphism between two substructures of $\bm{A}$.)  
We denote the set of all subisomorphisms of $\bm{A}$ by
$\si\bm{A}$. 

The following semantic characterization of open-definability is
central to our study.
\begin{thm}[{\cite[Thm 3.1]{lemas_semanticos}}]
\label{thm:lema sem}Let $\bm{A}$ be a finite relational structure
and $T\subseteq A^{m}$. The following are equivalent: 
\begin{enumerate}
\item\label{la} $T$ is open-definable in $\bm{A}$. 
\item\label{lb} $T$ is preserved by all subisomorphisms $\gamma$ of $\bm{A}$.
\item\label{lc} $T$ is preserved by all subisomorphisms $\gamma$ of $\bm{A}$ with $\left|\dom\gamma\right|\leq m$.
\end{enumerate}
\end{thm} 

\begin{proof}
  The equivalence of (\ref{la}) and (\ref{lb}) is proved in \cite[Thm
  3.1]{lemas_semanticos}.  Certainly (\ref{lb}) implies (\ref{lc}), so
  we show that (\ref{lc}) implies (\ref{lb}). Let $\gamma$ be a subiso
  of $\mathbf{A}$ and let
  $\left\langle a_{1},\dots,a_{m}\right\rangle \in T\cap(\dom\gamma)^m$.
  Note that the restriction $\gamma|_{\{a_{1},\dots,a_{m}\}}$ is a
  subiso of $\mathbf{A}$. Thus,
  $\gamma(\bar{a})=\gamma|_{\{a_{1},\dots,a_{m}\}}(\bar{a})\in T$.
\end{proof}

\subsection{Encodings and Sizes}\label{sec:size}

As is customary when considering complexity questions, the \emph{size}
of an object is the length of a string over a finite alphabet encoding
the object. We assume fixed encodings for vocabularies, relations,
structures and formulas, and define the size of these objects
according to these encodings. For a set $S$, let $|S|$ be the number
of elements in $S$, and for a relational vocabulary $\tau$, let $|\tau|$
be the number of relational symbols in $\tau$.

We write $\size(ob)$ to denote the size of an object $ob$. Even though
we do not specify the encodings we assume the following equalities
throughout this note. Let $\tau$ be a relational vocabulary,
$\mathbf{A}$ a $\tau$-structure, $T\subseteq A^{m}$ and $\varphi$ a
first-order formula.

\begin{itemize}
\item
  $\size(\tau)=(\left|\tau\right|+\sum_{R\in\tau}\ar(R))\log|\tau|$,\footnote{When
    an expression involving $\log x$ does note make sense, read it as
    $\max\{\log x,1\}$.}\\[-2em]
\item
  $\size(\mathbf{A})=\size(\tau)+(\left|A\right|+\sum_{R\in\tau}\ar(R)\left|R^{\mathbf{A}}\right|)\log|A|$,\\[-2em]
\item
  $\size(\mathbf{A},T)=(\size(\mathbf{A})+m\left|T\right|)\log|A|$,\\[-2em]
\item
  $\size_{\tau}(\varphi)=relcount(\varphi)\log|\tau|+varcount(\varphi)\log(var^{\#}(\varphi))$.
\end{itemize}

\noindent
Here $relcount(\varphi)$ {[}$varcount(\varphi)${]} stands for the
number of occurrences of relation symbols {[}variables{]} in $\varphi$,
and $var^{\#}(\varphi)$ is the number of different variables occurring
in $\varphi$. Since a formula $\varphi$ is a $\tau$-formula for every
$\tau$ containing the relation symbols in $\varphi$, the encoding of
$\varphi$ (and thus its size) depend on which vocabulary we have in
mind for $\varphi$. Another assumption we make on the encodings is
that determining whether $\bar{a}\in R^{\mathbf{A}}$ can be computed
in time $O(\size(\mathbf{A}))$.

\section{Classical Complexity of Open-Definability\label{sec:Classical-Complexity}}

In what follows a \emph{graph} is a model $\bm{G}$ of the vocabulary
$\tau_{\mathrm{GRAPH}}=\{E\}$, with $E$ binary, and such that $E^{\bm{G}}$
is symmetric and irreflexive. We provide a reduction from the following
problem to prove our hardness result.

\begin{center}
\begin{tabular}{|lrp{7cm}|} \hline
\multirow{2}{*}{\vspace*{3pt}\framebox{$\mathrm{InducedPath}$}}
& & \\
& \emph{Instance}:& A finite graph $\bm{G}$ and a positive
integer $k$.\\

& \emph{Question}:& Does $\bm{G}$ have a path of length $k$
as an induced subgraph (i.e., as a submodel)?\\[1em] \hline
\end{tabular}
\end{center}

\noindent
$\mathrm{InducedPath}$ is known to be $\mathrm{NP}$-complete (see, e.g.,~\cite{LibroNP}). 

\begin{thm}
\label{thm:OPENDEF is coNP complete}$\mathrm{OpenDef}$ is $\mathrm{coNP}$-complete.
\end{thm}

\begin{proof}
We first prove hardness. Fix an input graph $\bm{G}$ and
a positive integer $k$. We may assume that $G$ is disjoint with
the set of integers. Suppose first that $k=2l$. Let $\bm{G}'$ be
the graph with universe $G':=G\cup\{-l,\dots,-1,1,\dots,l\}$ and
with 
\[
E^{\bm{G}'}:=E^{\bm{G}}\cup\{\left\langle a,b\right\rangle \in\{-l,\dots,-1,1,\dots,l\}^{2}:\left|a-b\right|=1\}\cup\{\left\langle -1,1\right\rangle ,\left\langle 1,-1\right\rangle \}.
\]
That is, $\bm{G}'$ is the disjoint union of $\bm{G}$ and a path
of length $k$. Define 
\[
T:=\{\left\langle -l,\dots,-1,1,\dots,l\right\rangle ,\left\langle l,\dots,1,-1,\dots,-l\right\rangle \}.
\]
Now, observe that by Theorem~\ref{thm:lema sem} we have that
$\mathrm{OpenDef}$ returns $\mathrm{FALSE}$ on input $(\bm{G}',T)$ if
and only if $\mathrm{InducedPath}$ returns $\mathrm{TRUE}$ on input
$(\bm{G},k)$. The case where $k$ is odd is analogous.

Showing that $\mathrm{OpenDef}$ is in $\mathrm{coNP}$ is a straightforward
application of Theorem~\ref{thm:lema sem}. Given a finite relational
structure $\bm{A}$ and $T\subseteq A^{k}$, the fact that $\mathrm{OpenDef}$
returns $\mathrm{FALSE}$ on input $(\bm{A},T)$ is witnessed by a
bijection $\gamma$ between subsets of $A$ satisfying conditions
easily checked in poly-time with respect to the size of $(\bm{A},T)$.
\end{proof}
Given a relational signature $\tau$ let $\mathrm{OpenDef}[\tau]$
be the restriction of $\mathrm{OpenDef}$ to input structures of signature
$\tau$. In view of the proof of Theorem \ref{thm:OPENDEF is coNP complete}
we have the following.

\begin{cor}
$\mathrm{OpenDef}[\tau_{\mathrm{GRAPH}}]$ is $\mathrm{coNP}$-complete.
\end{cor}

\section{Parameterized Complexity of
  $\mathrm{OpenDef}$\label{sec:param}}

\emph{Parameterized complexity} is a mathematical framework that
allows for a more fine-grained analysis of the computational costs of
a problem than classical complexity. In a \emph{parameterization }of a
(classical) problem we single out a specific part of the input of the
problem to try and understand how this part affects the computational
cost.  For example, a parameterization of propositional $\mathrm{SAT}$
(i.e., the satisfiability problem for Propositional Logic) could be
the number of variables in the input formula.

We begin with the basic definitions involved. There are slight
discrepancies for these definitions in the literature, but our results
remain valid regardless of which of the versions is used. We follow
the account in~\cite{FlumGroheBook}. As is customary, (classical)
decision problems are formalized as languages over finite nonempty
alphabets. Let $\Sigma\neq\emptyset$ be a finite alphabet.
\begin{itemize}
\item A \emph{parameterization} of $\Sigma^{*}$ is a mapping $\kappa:\Sigma^{*}\rightarrow\mathbb{N}$
that is polynomial time computable.
\item A \emph{parameterized (or parametric) problem} (over $\Sigma$) is
a pair $(Q,\kappa)$ consisting of a set $Q\subseteq\Sigma^{*}$ of
strings over $\Sigma$ and a parameterization $\kappa$ of $\Sigma^{*}$.
\end{itemize}
We consider the following parameterization\footnote{This parameterized
  problem (and others appearing below) is presented in an informal
  way, but it should be clear how to cast them in the form of our formal
  definition. } of $\mathrm{OpenDef}$:

\medskip{}

\begin{center}
\begin{tabular}{|lrp{7cm}|} \hline
\multirow{4}{*}{\vspace*{25pt}\framebox{$p$\mhyphen$\mathrm{OpenDef}$}}
& & \\
&\emph{Instance}: & A finite relational structure $\bm{A}$
and $T\subseteq A^{m}$.\\

& \emph{Parameter}: & $m\left|T\right|$.\\

& \emph{Question}: & Is $T$ open-definable in $\bm{A}$?\\[1em] \hline
\end{tabular}
\end{center}

For a positive integer $k$, the $k$-th \emph{slice} of a parameterized
problem $(Q,\kappa)$ is the restriction of the problem to all
instances $x\in\Sigma^{*}$ such that $\kappa(x)=k$. Let
$p\mhyphen\mathrm{OpenDef}_{k}$ denote the $k$-th slice of
$p\mhyphen\mathrm{OpenDef}$.

\begin{prop}
  \label{prop:opendef slice}$p\mhyphen\mathrm{OpenDef}_{k}$ is
  computable in time $k!n^{2k}p(n)$ where $n$ is the size of the input
  and $p(X)$ a polynomial.
\end{prop}

\begin{proof}
Fix a finite relational structure $\bm{A}$ and $T\subseteq A^{m}$.
Let $n$ be the size of $(\bm{A},T)$ and $k=m\left|T\right|$. Given
a positive integer $l$ let 
\[
\si_{l}\bm{A}:=\{\gamma\in\si\bm{A}:\left|\dom\gamma\right|=l\}.
\]
Note that Corollary~\ref{thm:lema sem} implies that $T$ is open-definable
in $\bm{A}$ if and only if every $\gamma\in(\bigcup_{l\leq m}\si_{l}\bm{A})$
preserves $T$. We show that the right-hand side of this equivalence
can be checked in polynomial time. Let 
\[
\mathcal{I}_{l}:=\{\gamma:\text{there are }B,B'\subseteq A\text{ such that }\left|B\right|=l\text{ and }\gamma:B\rightarrow B'\text{ is bijective}\}.
\]
Observe that 
\[
\left|\mathcal{I}_{l}\right|=l!\cdot\binom{\left|A\right|}{l}^{2},
\]
since there are $l!$ bijections between any two subsets of size $l$
of $A$ , and so 
\[
\left|\mathcal{I}_{l}\right|\leq l!\cdot\left|A\right|^{2l}\leq k!\cdot n^{2k}.
\]
Now for each $\gamma\in\mathcal{I}_{l}$ we have to check: 
\begin{enumerate}
\item if $\gamma\in\si\bm{A}$, 
\item and in that case, if $\gamma$ preserves $T$. 
\end{enumerate}
Clearly both these tasks can be carried out in time bounded by a polynomial
in $n$. If $p_{l}(n)$ is such a polynomial, then checking if every
member of $\si_{l}\bm{A}$ preserves $T$ takes at most $\left|\mathcal{I}_{l}\right|\cdot p_{l}(n)$
steps. Thus, the computation for $\mathrm{OpenDef}_{k}$ on input
$(\bm{A},T)$ can be done in at most 
\[
k!n^{2k}\sum_{l=1}^{k}p_{l}(n)
\]
steps.
\end{proof}

A parameterized problem $(Q,\kappa)$ over the alphabet $\Sigma$
is \emph{fixed parameter tractable} (FPT) if there is an algorithm
$\mathcal{A}$ together with a polynomial $p(X)$ and a computable
function $f:\mathbb{N}\rightarrow\mathbb{N}$ such that $\mathcal{A}$
decides if $x\in Q$ in time $f(\kappa(x))p(\left|x\right|)$ for
a all $x\in\Sigma^{*}$. The class $\mathrm{FPT}$ of all fixed parameter
tractable problems plays the role $\mathrm{P}$ plays in classical
complexity.

Even though each slice of $p\mhyphen\mathrm{OpenDef}$ can be computed
in polynomial time, the bound given by Proposition~\ref{prop:opendef
  slice} does not imply
$p\mhyphen\mathrm{OpenDef}\in\mathrm{FPT}$, since the parameter
appears as an exponent of the size of the input.  In fact, as we shall
see below it is unlikely that $p\mhyphen\mathrm{OpenDef}$ is FPT,
since it is hard for the class $\mathrm{coW}[1]$; a class of
parameterized problems believed to be strictly larger than
$\mathrm{FPT}$.  But before we can discuss hardness of parameterized
problems we need an adequate notion of reduction.

\begin{defn}
  \label{def:fpt-reduction}Let $(Q,\kappa)$ and $(Q',\kappa')$ be
  parameterized problems over the alphabets $\Sigma$ and $\Sigma^{*}$,
  respectively. An \emph{fpt-reduction} from $(Q,\kappa)$ to
  $(Q',\kappa')$ is a mapping
  $R:\Sigma^{*}\rightarrow(\Sigma')^{*}$such that:
\begin{enumerate}
\item For all $x\in\Sigma^{*}$ we have $(x\in Q\Leftrightarrow R(x)\in Q')$.
\item There is a computable function $f$ and a polynomial $p(X)$ such
that $R(x)$ is computable in time $f(\kappa(x))\cdot p(x)$.
\item There is a computable function $g:\mathbb{N}\rightarrow\mathbb{N}$
such that $\kappa(R(x))\leq g(\kappa(x))$ for all $x\in\Sigma^{*}$.
\end{enumerate}
If $P$ and $P'$ are parameterized problems, we write $P\leq^{\textrm{fpt}}P'$
if there is an fpt-reduction from $P$ to $P'$, and write $P\equiv^{\mathrm{fpt}}P'$
if there are fpt-reductions in both directions.
\end{defn}

There are other notions of fpt-reduction, such as Turing
fpt-reductions~\cite{FlumGroheBook}, involving oracles. All
fpt-reductions in this note satisfy
Definition~\ref{def:fpt-reduction}, and thus we simply use the name
fpt-reduction for them.

To analyze the complexity of parametric problems which appear not to
be tractable, Downey and Fellows introduced the $\mathrm{W}$ hierarchy
\cite{Downey_Fellows_W_hierarchy}. The classes 
\[
\mathrm{FPT}\subseteq\mathrm{W}[1]\subseteq\mathrm{W}[2]\subseteq\dots\subseteq\mathrm{W}[P]
\]
in this hierarchy are closed under fpt-reductions and are believed to
be all different. They have many natural complete problems (see,
e.g.,~\cite{DowneyFellowsBook,FlumGroheBook}).  One can think of this hierarchy as
analogous to the polynomial hierarchy in classical complexity. We only
consider the classes $\mathrm{W}[1]$ and $\mathrm{W}[P]$ in the
sequel. Their formal definitions are somewhat involved and not needed
in our arguments, so we do not include them here. (The interested
reader can find them in \cite{FlumGroheBook}.)  We do need the
following characterization of $\mathrm{W}[P]$, analogous to the
characterization of $\mathrm{NP}$ in terms of ``certificates''.

\begin{lem}[{\cite[Lem 3.8]{FlumGroheBook}}]
  A parameterized problem $(Q,\kappa)$ over the alphabet $\Sigma$ is
  in $\mathrm{W}[P]$ if and only if there are computable functions
  $f,h:\mathbb{N}\rightarrow\mathbb{N}$, a polynomial $p(X)$, and a
  $Y\subseteq\Sigma^{*}\times\{0,1\}^{*}$ such that:
  \begin{enumerate}
  \item For all $(x,y)\in\Sigma^{*}\times\{0,1\}^{*}$ it is decidable
    in time $f(\kappa(x))p(\left|x\right|)$ whether $(x,y)\in Y$.
  \item For all $(x,y)\in\Sigma^{*}\times\{0,1\}^{*}$, if $(x,y)\in Y$
    then $\left|y\right|=h(\kappa(x))\log\left|x\right|$.
  \item For every $x\in\Sigma^{*}$ we have $x\in Q$ iff there exists
    $y\in\{0,1\}^{*}$ such that $(x,y)\in Y$.
  \end{enumerate}
\end{lem}

The \emph{complement} $(Q,\kappa)^{\mathsf{C}}$ of a parametric problem
$(Q,\kappa)$ is the parametric problem $(Q\backslash\Sigma^{*},\kappa)$.
It follows directly from the definitions that $P_{1}\leq^{\textrm{fpt}}P_{2}$
iff $P_{1}^{\mathsf{C}}\leq^{\textrm{fpt}}P_{2}^{\mathsf{C}}$. For
a class $\mathrm{K}$ of parametric problems, let $\mathrm{coK}$
denote the class off all parametric problems whose complement is in
$\mathrm{K}$.

\begin{prop}
$p\mhyphen\mathrm{OpenDef}\in\mathrm{coW}[P]$.
\end{prop}

\begin{proof}
By Theorem \ref{thm:lema sem} we know that $T\subseteq A^{m}$ fails
to be open-definable in $\mathbf{A}$ if and only if there is a bijection
$\gamma$ between subsets of $A$ of cardinality at most $m$, satisfying
that $\gamma$ is a subisomorphism of $\mathbf{A}$ and does not preserve
$T$. Such a bijection can be encoded as a binary string of length
$O(m^{2}\log\left|A\right|)$, and we can compute in time polynomial
in $\size(\mathbf{A},T$) if $\gamma$ is a subisomorphism of $\mathbf{A}$
not preserving $T$. (Note that this is a refinement of the argument
we used in Theorem \ref{thm:OPENDEF is coNP complete} to prove $\mathrm{OpenDef}\in\mathrm{coNP}$.)
\end{proof}

Next we establish a lower bound for the complexity of
$p\mhyphen\mathrm{OpenDef}$ by a reduction from the following
parameterized version of the Clique problem.

\begin{center}
\begin{tabular}{|lrp{8cm}|} \hline
\multirow{3}{*}{\vspace*{15pt}\framebox{$p\mhyphen\mathrm{Clique}$}}
& & \\
&\emph{Instance}: & A finite graph $\bm{G}$ and a positive
integer $k$.\\

&\emph{Parameter}:& $k$.\\

& \emph{Question}:& Does $\bm{G}$ have clique of size $k$?\\[1em] \hline
\end{tabular}
\end{center}

It is proved in \cite[Cor 3.2]{DOWNEY1995_Clique_W1_Complete} that
$p\mhyphen\mathrm{Clique}$ is complete (under fpt-reductions) for
the class $\mathrm{W}[1]$.
\begin{lem}
\label{lem:pclique se reduce a podc}$p\mhyphen\mathrm{Clique}\leq^{\textrm{fpt}}p\mhyphen\mathrm{OpenDef}^{\mathsf{C}}$;
hence $p\mhyphen\mathrm{OpenDef}$ is hard for $\mathrm{coW}[1]$.
\end{lem}

\begin{proof}
The idea is the same as in the proof of Theorem \ref{thm:OPENDEF is coNP complete}.
Given an input $\mathbf{G},k$ for $p\mhyphen\mathrm{Clique}$ the
reduction computes the input $\mathbf{G}\sqcup\mathbf{K}_{k},T_{k}$
for $p\mhyphen\mathrm{OpenDef}^{\mathsf{C}}$, where $\mathbf{G}\sqcup\mathbf{K}_{k}$
is the disjoint union of $\mathbf{G}$ with the complete graph on
the vertices $\{1,\dots,k\}$, and $T_{k}=\{(\sigma(1),\dots,\sigma(k)):\sigma\text{ a permutation of }\{1,\dots,k\}\}$.
It is easy to see that this is an fpt-reduction, and that $\mathbf{G}$
has a clique of size $k$ iff $T_{k}$ is not open-definable in $\mathbf{G}\sqcup\mathbf{K}_{k}$.
(Note that this is \emph{not }a polynomial reduction.)
\end{proof}

In contrast with our analysis of the classical complexity of
$\mathrm{OpenDef}$, we were not able to show that
$p\mhyphen\mathrm{OpenDef}$ is in $\mathrm{coW}[1]$.  However when we
fix the vocabulary we can establish a sharp upper bound. For a
vocabulary $\tau$ let $p\mhyphen\mathrm{OpenDef}[\tau]$ denote the
restriction of $p\mhyphen\mathrm{OpenDef}$ to input structures with
vocabulary $\tau$, and let
$p\mhyphen\mathrm{OpenDef}^{\mathsf{C}}[\tau]$ denote the complement
of this problem. Before we start on the upper bounds, we have the
following consequence of Theorem \ref{lem:pclique se reduce a podc}.

\begin{cor}
  \label{cor:.pclique reduce a pODc=00005Bt=00005D}For every
  vocabulary $\tau$ with at least one at least binary relation we have
  $p\mhyphen\mathrm{Clique}\leq^{\textrm{fpt}}p\mhyphen\mathrm{OpenDef}^{\mathsf{C}}[\tau]$;
  and hence $p\mhyphen\mathrm{OpenDef}[\tau]$ is hard for
  $\mathrm{coW}[1]$.
\end{cor}

\begin{proof}
If $\tau$ has at least one at least binary relation, it is easy to
see that 
\[
p\mhyphen\mathrm{OpenDef}^{\mathsf{C}}[\tau_{\mathrm{GRAPH}}]\leq^{\textrm{fpt}}p\mhyphen\mathrm{OpenDef}^{\mathsf{C}}[\tau].
\]
From the proof of Lemma~\ref{lem:pclique se reduce a podc} it follows 
that $p\mhyphen\mathrm{Clique}\leq^{\textrm{fpt}}p\mhyphen\mathrm{OpenDef}^{\mathsf{C}}[\tau_{\mathrm{GRAPH}}]$.
\end{proof}

We turn now to establishing an upper bound for
$p\mhyphen\mathrm{OpenDef}[\tau]$.  Recall that a sentence is
existential if it has the form
$\exists v_{1}\dots\exists v_{l}\,\alpha(v_{1},\dots,v_{l})$ where
$\alpha$ is open. Let $\Sigma_1[\tau]$ be the set of all existential
sentences over a vocabulary $\tau$. For a given vocabulary $\tau$,
consider the following parameterized model checking problem:

\begin{center}
\begin{tabular}{|lrp{7cm}|} \hline
\multirow{4}{*}{\vspace*{25pt}\framebox{$p\mhyphen\mathrm{MC}(\Sigma_{1}[\tau])$}}
& & \\

& \emph{Instance}: & A finite $\tau$-structure $\mathbf{A}$
and an existential $\tau$-sentence $\varphi$.\\

& \emph{Parameter}:& $\size_{\tau}(\varphi)$.\\

& \emph{Question}: & Does $\mathbf{A}$ satisfy $\varphi$?\\[1em] \hline
\end{tabular}
\end{center}

Recall that
$\size_{\tau}(\varphi)=relcount(\varphi)\log|\tau|+varcount(\varphi)\log(var^{\#}(\varphi))$.
It is proved in~\cite{FlumGrohePaper} that
$p\mhyphen\mathrm{MC}(\Sigma_{1}[\tau])$ is in $\mathrm{W}[1]$ for all
$\tau$.

\begin{thm}
\label{thm:pODc=00005Bt=00005D reduce a pMCE=00005Bt=00005D}For every
vocabulary $\tau$, $p\mhyphen\mathrm{OpenDef}^{\mathsf{C}}[\tau]\leq^{\textrm{fpt}}p\mhyphen\mathrm{MC}(\Sigma_{1}[\tau])$.
Thus, $p\mhyphen\mathrm{OpenDef}[\tau]\in\mathrm{coW}[1]$.
\end{thm}

\begin{proof}
Fix a $\tau$-structure $\mathbf{A}$ and $T\subseteq A^{m}$. For
each tuple $\bar{a}=\left\langle a_{1},\dots,a_{m}\right\rangle \in A^{m}$,
and each $R\in\tau$ let $\Delta_{\bar{a},R}(x_{1},\dots,x_{m})$
be the the conjunction of the following set of atomic formulas
\[
\{R(x_{i_{1}},\dots,x_{i_{r}}):(a_{i_{1}},\dots,a_{i_{r}})\in R^{\mathbf{A}}\}\cup\{\neg R(x_{i_{1}},\dots,x_{i_{r}}):(a_{i_{1}},\dots,a_{i_{r}})\notin R^{\mathbf{A}}\},
\]
where $r$ is the arity of $R$. Observe that 
\begin{align*}
\size(\Delta_{\bar{a},R}) & =m^{r}\log\left|\tau\right|+rm^{r}\log m\\
 & \leq q_{r}(m)
\end{align*}
for a suitable polynomial $q_{r}(X)$. Also observe that $\Delta_{\bar{a},R}$
can be computed in time $O(m^{r}\size(\mathbf{A})+\size(\Delta_{\bar{a},R}))$,
so there is a polynomial $p_{r}(X)$ such that the computation of
$\Delta_{\bar{a},R}$ can be done in at most $p_{r}(m)\size(\mathbf{A})$
steps. Next, define 
\[
\Delta_{\bar{a}}(x_{1},\dots,x_{m}):=\bigwedge\{\Delta_{\bar{a},R}(x_{1},\dots,x_{m}):R\in\tau\}.
\]
Let $\rho$ be the greatest among the arities of the relations in
$\tau$. Then, $\size(\Delta_{\bar{a}})\leq\left|\tau\right|p_{\rho}(m)$,
and $\Delta_{\bar{a}}$ is computable in time bounded by $\left|\tau\right|p_{\rho}(m)\size(\mathbf{A})$.
Note that $\Delta_{\bar{a}}$ characterizes the isomorphism type of
$\bar{a}$ in $\mathbf{A}$, i.e., for all $\bar{b}\in A^{m}$ we
have
\[
\mathbf{A}\vDash\Delta_{\bar{a}}[b_{1},\dots,b_{m}]\Longleftrightarrow\bar{a}\mapsto\bar{b}\text{ is a subisomorphism of }\mathbf{A}.
\]
Now, let

\[
\Delta_{T}(x_{1},\dots,x_{m}):=\bigvee\{\Delta_{\bar{a}}(\bar{x}):\bar{a}\in T\},
\]
and take $\varphi_{\mathbf{A},T}$ as the sentence

\[
\exists\bar{x}_{1}\dots\exists\bar{x}_{t+1}\,\,\bigwedge\{\bar{x}_{i}\neq\bar{x}_{j}:1\leq i<j\leq t+1\}\wedge\bigwedge\{\Delta_{T}(\bar{x}_{i}):1\leq i\leq t+1\},
\]
where $t$ is the number of tuples in $T$. It is straightforward
to check that there are polynomials $p(X)$ and $q(X)$ such that:
\begin{enumerate}
\item[(1)] $\varphi_{\mathbf{A},T}$ can be computed in time $p(m\left|T\right|)\size(\mathbf{A})$,
\end{enumerate}
and
\begin{enumerate}
\item[(2)] $\size(\varphi_{\mathbf{A},T})\leq q(m\left|T\right|).$
\end{enumerate}
Next, note that $\mathbf{A}\vDash\Delta_{T}[\bar{b}]$ if and only
if $\bar{b}$ has the same isomorphism type as some tuple in $T$;
thus $\varphi_{\mathbf{A},T}$ asserts that there are $t+1$ distinct
$m$-tuples such that each one has the same isomorphism type as some
tuple in $T$. So, $\mathbf{A}\vDash\varphi_{\mathbf{A},T}$ if and
only if there are $\bar{a}\in T$ and $\bar{b}\in A^{m}\setminus T$
such that $\bar{a}\mapsto\bar{b}$ is a subisomorphism of $\mathbf{A}$.
By Theorem \ref{thm:lema sem}, this says that:
\begin{enumerate}
\item[(3)]  $\mathbf{A}\vDash\varphi_{\mathbf{A},T}$ iff $T$ is not open-definable
in $\mathbf{A}$. 
\end{enumerate}

To conclude, observe that (1-3) guarantee that the transformation
$\mathbf{A},T\rightsquigarrow\mathbf{A},\varphi_{\mathbf{A},T}$ is an
fpt-reduction from $p\mhyphen\mathrm{OpenDef}^{\mathsf{C}}[\tau]$ to
$p\mhyphen\mathrm{MC}(\Sigma_{1}[\tau])$. (Notably, it is also a
polynomial many-one reduction.)
\end{proof}

It is worth to note that, by Corollary, Theorem
\ref{thm:pODc=00005Bt=00005D reduce a pMCE=00005Bt=00005D} actually
holds with $\equiv^{\mathrm{fpt}}$ in place of $\le^{\mathrm{fpt}}$.
That is,
\[
p\mhyphen\mathrm{OpenDef}^{\mathsf{C}}[\tau]\equiv^{\mathrm{fpt}}p\mhyphen\mathrm{MC}(\Sigma_{1}[\tau])
\]
 for every vocabulary $\tau$ with at least one at least  binary relation.

Combining the upper an lower bounds found above we obtain the main
result of this section.

\begin{thm}
$p\mhyphen\mathrm{OpenDef}[\tau]$ is $\mathrm{coW}[1]$-complete
for every vocabulary $\tau$ with at least one at least binary relation.
In particular, $p\mhyphen\mathrm{OpenDef}[\tau_{\mathrm{GRAPH}}]$
is $\mathrm{coW}[1]$-complete.
\end{thm}

\begin{proof}
Combine Corollary \ref{cor:.pclique reduce a pODc=00005Bt=00005D}
and Theorem \ref{thm:pODc=00005Bt=00005D reduce a pMCE=00005Bt=00005D}.
\end{proof}

\section{The Length of Open Formulas\label{sec:LargoFormulas}}

In previous sections we show that the definability problem for open
formulas is not tractable. We now discuss the \emph{size} of formulas
required in definitions. In particular, we show that it is not possible
to polynomially bound the size of an open formula required in the
definition of a relation in a given structure. More formally, we construct
a sequence $\{(\mathbf{A}_{n},T_{n}):n\in\mathbb{N}\}$ whose size
grows polynomially in $n$, and such that the smallest definition
of $T_{n}$ in $\mathbf{A}_{n}$ is exponentially large in $n$. This
does not come as a surprise though; a polynomial bound on the size
of defining open formulas would entail that $\mathrm{OpenDef}$ is
in $\mathrm{NP}$ and, since we know $\mathrm{OpenDef}$ is $\mathrm{coNP}$-complete
(Theorem \ref{thm:OPENDEF is coNP complete}), we would have $\mathrm{coNP}\subseteq\mathrm{NP}$.
\begin{thm}
For each $n\geq3$ there are a finite relational structure $\mathbf{A}_{n}$
and an $n^{2}$-ary relation $T_{n}$ over $A_{n}$ such that:
\begin{itemize}
\item $\size(\mathbf{A}_{n},T_{n})=O(n^{3}\log n)$, 
\item $T_{n}$ is open-definable in $\mathbf{A}_{n}$, and
\item every open formula defining $T_{n}$ in $\mathbf{A}_{n}$ has at least
$(n-1)^{n}$ literals.
\end{itemize}
\end{thm}

\begin{proof}
Fix a natural number $n\geq3$ and let $A$ be the union of four pairwise
disjoint $n$-element sets, say
\[
A=\{a_{1},\dots,a_{n}\}\cup\{b_{1},\dots,b_{n}\}\cup\{c_{1},\dots,c_{n}\}\cup\{\ast_{1},\dots,\ast_{n}\}.
\]
Let $M_{1}$ be the $n\times n$ matrix such that
\begin{itemize}
\item the first column of $M_{1}$ is $\bar{a}$, and the rest of its entries
are $\ast_{1}$.
\end{itemize}
For each $j\in\{2,\dots,n\}$ take $M_{j}$ to be the $n\times n$
matrix such that
\begin{itemize}
\item the first column of $M_{j}$ is $\bar{b}$, the $j$-th column of
$M_{j}$ is $\bar{c}$, and the rest of its entries are $\ast_{j}$.
\end{itemize}
Next, define
\[
R:=\{\bar{a},\bar{b},\bar{c}\}\cup\{\text{rows of }M_{1}\}\cup\{\text{rows of }M_{2}\}\cup\dots\cup\{\text{rows of }M_{n}\}.
\]
In what follows we identify $n^{2}$-tuples with $n\times n$ matrices
(in the obvious way). We adjust the indexes of our variables accordingly,
i.e., we shall write 
\[
\varphi(\begin{array}{c}
x_{11}\dots x_{1n}\\
\ddots\\
x_{n1}\dots x_{nn}
\end{array})
\]
instead of $\varphi(x_{1},\dots,x_{n^{2}})$.

Let $T=\{M_{1}\}$; we prove that $\mathbf{A}_{n}:=(A,R)$ and $T_{n}:=T$
satisfy the conditions of the theorem. First observe that, according
to our definition of size (see Section~\ref{sec:size}), we have 
\begin{align*}
\size(\mathbf{A}_{n},T_{n}) & =(\size(\mathbf{A}_{n})+n^{2})\log4n\\
 & =(1+n)+(4n+n(3+n^{2})+n^{2})\log4n\\
 & =O(n^{3}\log n)
\end{align*}
 Next, define 
\[
\alpha(\begin{array}{c}
x_{11}\dots x_{1n}\\
\ddots\\
x_{n1}\dots x_{nn}
\end{array}):=R(x_{11},\dots,x_{n1})\wedge\bigwedge_{j=1}^{n}R(x_{1j},\dots,x_{nj}).
\]
Suppose $M$ is an $n\times n$ matrix with entries in $A$ such that
$\mathbf{A}_{n}\vDash\alpha(M)$. Then the first column of $M$, say
$\bar{m}=(m_{1},\dots,m_{n})$, and each of its rows must be in $R$.
As each $m_{i}$ must be the first coordinate of some tuple in $R$,
we have that $\bar{m}\in\{\bar{a},\bar{b}\}$. If $\bar{m}=\bar{a}$,
then $M=M_{1}$, since for each $i\in\{1,\dots,n\}$ there is exactly
one tuple in $R$ with first coordinate $a_{i}$. On the other hand,
if $\bar{m}=\bar{b}$, we have that for each $i\in\{1,\dots,n\}$
the $i$-th row of $M$ agrees with the $i$-th row of $M_{j_{i}}$
for some $j_{i}\in\{2,\dots,n\}$.

Given $\bar{j}\in\{2,\dots,n\}^{n}$ let $M^{\bar{j}}$ be the $n\times n$
matrix whose $i$-th row is the $i$-th row of $M_{j_{i}}$ for $i\in\{1,\dots,n\}$.
From the considerations in the last paragraph it follows that 
\[
\{M:\mathbf{A}_{n}\vDash\alpha(M)\}=\{M_{1}\}\cup\{M^{\bar{j}}:\bar{j}\in\{2,\dots,n\}^{n}\}.
\]
For each $\bar{j}\in\{2,\dots,n\}^{n}$ define
\[
\lambda_{\bar{j}}(\begin{array}{c}
x_{11}\dots x_{1n}\\
\ddots\\
x_{n1}\dots x_{nn}
\end{array}):=R(x_{1j_{1}},\dots,x_{nj_{n}}),
\]
and observe that $\mathbf{A}_{n}\vDash\lambda_{\bar{j}}(M^{\bar{j}})$.
So, if we take 
\[
\beta(\begin{array}{c}
x_{11}\dots x_{1n}\\
\ddots\\
x_{n1}\dots x_{nn}
\end{array}):=\bigwedge_{\bar{j}\in\{2,\dots,n\}^{n}}\neg\lambda_{\bar{j}},
\]
it follows that $\mathbf{A}_{n}\vDash(\alpha\wedge\beta)(M)\Longleftrightarrow M=M_{1}$.
Thus $T=\{M_{1}\}$ is open-definable in $\mathbf{A}_{n}$.

To conclude, note that $\lambda_{\bar{j}}$ is the only atomic formula
that distinguishes $M_{1}$ from $M^{\bar{j}}$. Thus, if $\varphi$
is an open formula that defines $T$, then $\neg\lambda_{\bar{j}}$
must occur in $\varphi$ for every $\bar{j}\in\{2,\dots,n\}^{n}$.
\end{proof}
\bibliographystyle{plain}
\bibliography{opendef}

\end{document}